\documentclass[12pt]{article}
\usepackage{graphicx}
\def\hybrid{\topmargin 0pt      \oddsidemargin 0pt
        \headheight 0pt \headsep 0pt
       \voffset-1cm
        \textwidth 6.25in       % A4 paper
       \textheight 9.5in       % A4 paper
        \marginparwidth 0.0in
        \parskip 5pt plus 1pt   \jot = 1.5ex}
\catcode`\@=11
\def\marginnote#1{}

\newcount\hour
\newcount\minute
\newtoks\amorpm
\hour=\time\divide\hour by60
\minute=\time{\multiply\hour by60 \global\advance\minute by-\hour}
\edef\standardtime{{\ifnum\hour<12 \global\amorpm={am}%
        \else\global\amorpm={pm}\advance\hour by-12 \fi
        \ifnum\hour=0 \hour=12 \fi
        \number\hour:\ifnum\minute<10 0\fi\number\minute\the\amorpm}}
\edef\militarytime{\number\hour:\ifnum\minute<10 0\fi\number\minute}

\def\draftlabel#1{{\@bsphack\if@filesw {\let\thepage\relax
   \xdef\@gtempa{\write\@auxout{\string
      \newlabel{#1}{{\@currentlabel}{\thepage}}}}}\@gtempa
   \if@nobreak \ifvmode\nobreak\fi\fi\fi\@esphack}
        \gdef\@eqnlabel{#1}}
\def\@eqnlabel{}
\def\@vacuum{}
\def\draftmarginnote#1{\marginpar{\raggedright\scriptsize\tt#1}}

\def\draftlabel#1{{\@bsphack\if@filesw {\let\thepage\relax
   \xdef\@gtempa{\write\@auxout{\string
      \newlabel{#1}{{\@currentlabel}{\thepage}}}}}\@gtempa
   \if@nobreak \ifvmode\nobreak\fi\fi\fi\@esphack}
        \gdef\@eqnlabel{#1}}
\def\@eqnlabel{}
\def\@vacuum{}
\def\draftmarginnote#1{\marginpar{\raggedright\scriptsize\tt#1}}

\def\draft{\oddsidemargin -.5truein
        \def\@oddfoot{\sl preliminary draft \hfil
        \rm\thepage\hfil\sl\today\quad\militarytime}
        \let\@evenfoot\@oddfoot \overfullrule 3pt
        \let\label=\draftlabel
        \let\marginnote=\draftmarginnote
   \def\@eqnnum{(\theequation)\rlap{\kern\marginparsep\tt\@eqnlabel}%
\global\let\@eqnlabel\@vacuum}  }

\def\numberbysection{\@addtoreset{equation}{section}
        \def\theequation{\thesection.\arabic{equation}}}

\def\underline#1{\relax\ifmmode\@@underline#1\else
        $\@@underline{\hbox{#1}}$\relax\fi}

\def\titlepage{\@restonecolfalse\if@twocolumn\@restonecoltrue\onecolumn
     \else \newpage \fi \thispagestyle{empty}\c@page\z@
        \def\thefootnote{\fnsymbol{footnote}} }

\def\endtitlepage{\if@restonecol\twocolumn \else  \fi
        \def\thefootnote{\arabic{footnote}}
        \setcounter{footnote}{0}}  %\c@footnote\z@ }
\relax

\hybrid

\newfont{\Bbb}{msbm10 scaled 1\@ptsize00}
\newfont{\Bbbb}{msbm7 scaled 1\@ptsize00}
\newcommand{\CC}{\mbox{\Bbb C}}
\newcommand{\CCC}{\mbox{\Bbbb C}}

\newcommand{\DDD}{\raise-1pt\hbox{$\mbox{\Bbbb D}$}}

\newcommand{\RR}{\mbox{\Bbb R}}

\newcommand{\UUU}{\raise-1pt\hbox{$\mbox{\Bbbb U}$}}

\newcommand{\z}{\raise-1pt\hbox{$\mbox{\Bbbb Z}$}}

\newcommand{\sss}{\raise-1pt\hbox{$\mbox{\Bbbb S}$}}

\def\beq{\begin{equation}}
\def\eeq{\end{equation}}
\def\p{\partial}

\def\gr{{\bf grad \,}}
\def\Gr{{\bf Grad \,}}

\newtheorem{theorem}{Theorem}
\newtheorem{lemma-definition}{Lemma-Definition}[section]

\newtheorem{remark}{Remark}
\newtheorem{definition}{Definition}
\def\square{\hfill
{\vrule height6pt width6pt depth1pt} \break \vspace{.01cm}}

\begin{document}

\begin{titlepage}

\title{Gradient nature of Laplacian growth}

\author{
A.~Zabrodin\thanks{
National Research University Higher School of Economics,
20 Myasnitskaya Ulitsa,
Moscow 101000, Russia and
NRC ``Kurchatov institute'', Moscow, Russia;
e-mail: zabrodin@itep.ru}}

\date{September 2026}
\maketitle

\vspace{-7cm} \centerline{ \hfill ITEP-TH-34/26}\vspace{7cm}

%\begin{center}

\hfill{\it To Andrey Pogrebkov on his 80th birthday}

%\end{center}

\begin{abstract}

For a class of growth processes of Laplacian type in the plane, 
we suggest an interpretation as a ``gradient descent'' in the 
space of smooth closed curves. More precisely, we show that boundary of a
growing domain moves along a gradient of a certain functional
in the space of curves. In the simplest cases this functional 
is $\log (1/r)$, where $r$ is the external conformal radius 
of the growing domain.

\end{abstract}

\end{titlepage}

\vspace{5mm}

%

%\newpage
\tableofcontents

\vspace{5mm}

\section{Introduction}

Growth problems of Laplacian type such as Hele-Shaw viscous flows
refer to dynamics of a moving front (an interface)
between two distinct phases driven by a harmonic scalar field.
These essentially nonlinear
and non-local problems attract much attention
for quite a long time \cite{list}.
The Laplacian growth (LG) problem
appears in different physical and mathematical contexts
and has important practical applications.
For reviews see \cite{RMP,EV,book}.
In this paper, we shall have in mind the dynamics 
of an interface between two
incompressible fluids with very different viscosities.
In practice the 2D geometry is
realized in the Hele-Shaw cell -- a narrow gap between two parallel
plates (Fig.~\ref{figure:hele}). 

An informal description of the process in the cell is as follows.
(The precise mathematical formulation of the 
LG problem is given in the next section.)
A drop of an incompressible fluid with negligible viscosity 
(say, water) is surrounded by a viscous fluid (e.g., oil), 
which is drawn away from the edges of the cell 
(mathematically, this means that the ``pump'' is located at infinity).
The water drop ``grows'' over time due to the water 
entering through the thin tube at the center of the cell, 
increasing its surface area.
The law according to which the boundary of the drop grows 
is called Darcy's law.
It states that velocity at any
boundary point $\xi $ is proportional to the pressure gradient 
in the viscous fluid at that point\footnote{For simplicity, 
the filtration coefficient is put equal to 1,
cf. (\ref{darcy}).}:
\beq\label{int1}
{\bf v}(\xi ) = -\gr p(\xi ).
\eeq

If surface tension is negligibly
small, it can be considered that the pressure in oil 
is constant along the entire boundary. 
This case is seemingly simpler for theoretical analysis.
Apart from its (apparent) simplicity,
the case of vanishing surface tension is distinguished
in at least two respects.
\begin{itemize}
\item
There is a fairly large class of nontrivial 
exact solutions \cite{EV,book,exact}. 
However, many of them
``explode'' within a finite time, i.e., 
they develop singularities, after which the solutions loose 
any physical meaning. 
The singularities are regularized by introducing a surface
tension (though small but nonzero). However, the standard way
of introducing surface tension
breaks the nice mathematical
structure of the problem.
\item
In the limit of zero surface tension, the LG dynamics becomes 
highly unstable (i.e., any arbitrarily 
small disturbances lead, after some time, to significant 
consequences).
This instability leads to the fact 
(observed in experiments) 
that even if the drop initially had a circular shape, 
over time, the interface becomes increasingly meandering, 
with a lot of ``fingers'' and fjords, resembling the 
fractal structure of DLA clusters (see Fig. \ref{figure:fingers}). 
\end{itemize}

\begin{figure}[t]
\centering{\includegraphics[scale=0.3]{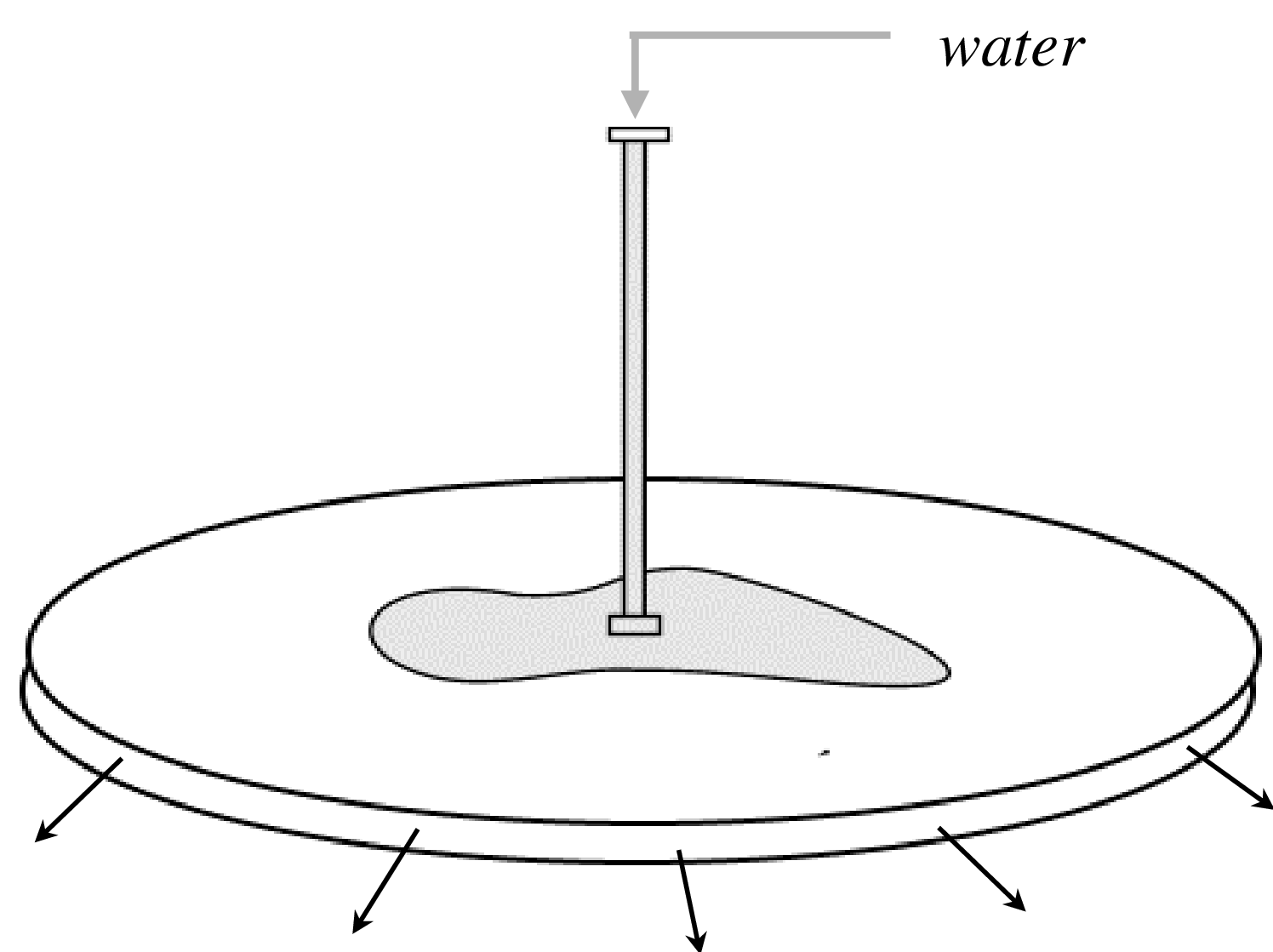}}
\vspace{1cm}
\caption{The Hele-Shaw cell.}
\label{figure:hele}
\end{figure}

\noindent
Thus, because of the instabilities,
the exact solutions hardly have physical meaning over 
long timescales. To adequately describe the dynamics 
at zero surface tension, 
taking into account random small 
fluctuations and perturbations seems inevitable.
Perhaps it is necessary to abandon the 
deterministic description and include, in a clever way, 
some stochastic elements (a noise) in the theory, passing to
a probabilistic description.
In other words,
it cannot be ruled out that the stochasticity  itself 
might play a role of a kind of ``regularization''
at zero surface tension.

\begin{figure}[t]
\centering{\includegraphics[scale=0.7]{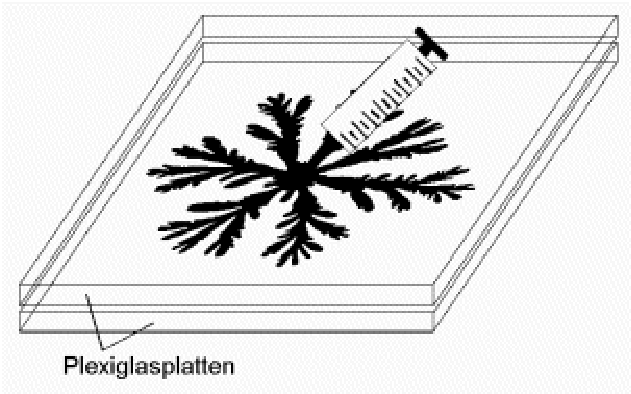}}
\vspace{1cm}
\caption{The fractal-like structure of the 
interface at large timescales.}
\label{figure:fingers}
\end{figure}

Remarkably, the
2D LG with zero surface tension has
a hidden integrable structure.
The first indication of this was obtained in Richardson’s work
\cite{Richardson}, where an infinite set of conserved quantities
was found for the problem in the radial geometry, as in 
Fig. \ref{figure:hele}.
The explanation for this and the true place 
of LG within the entire variety 
of integrable systems was obtained in the work
\cite{MWWZ}
and in the later works \cite{WZ00}-\cite{Z07}.
Specifically, a whole class of LG-like processes  
was embedded into an integrable
hierarchy of nonlinear partial differential equations which is
a zero dispersion version \cite{KriW,TakTak}
of the 2D Toda hierarchy \cite{UenoTakasaki}. 
Also, it should be noted that the
integrable structure unites the LG problem with 
important areas 
of mathematics and theoretical physics such as 
inverse potential problem, quadrature domains, random matrices,
theory of solitons and $c=1$ string theory. Some of these
links are reviewed in \cite{MWPT}.  
However, despite these wonderful mathematical structures, 
there is still no complete clarity on how to incorporate 
stochastic elements into the LG dynamics 
(and in a way that is somehow consistent with them),
and the problem is under discussion. See, e.g., the series
of works \cite{AMW16}--\cite{A25}, where consideration is based on 
the variational principle of maximum entropy, and also the 
close connection with models of random matrices is explored.

In this note, I would like to draw attention 
to a relatively simple property of the LG 
which is, however, not directly 
related to the integrable structures mentioned above. 
Nevertheless, this property could potentially pave a smart way 
for introducing a noise into the problem, as is mentioned above.
As far as I know, this possibility has not been discussed in 
the literature before.

The point is that, as shown in this note, 
the LG with zero surface tension 
can be understood as a ``gradient descent'' 
in a continually infinite-dimensional space of curves 
(contours, loops) on the plane. 
Namely, we show that the interface (boundary of the
growing domain) moves along a properly defined 
gradient of a certain functional
in the space of curves. In the simplest cases this functional 
is $\log (1/r)$, where $r$ is the external conformal radius 
of the growing domain:
\beq\label{int2}
{\bf V} =-\Gr \Bigl (\log (1/r)\Bigr ).
\eeq
Here $\Gr$ is the gradient in the space of 
loops and ${\bf V}$ is the ``velocity'' of the curve in that space
understood as the set $\{ {\bf v}(\xi )\}$ of local velocities 
of points of the curve. More precise definitions are given in 
Section 4. Comparing equations (\ref{int1}) and (\ref{int2}),
one can say that the latter is a kind of ``Darcy's law in the space
of loops''.

The paper is organized as follows. Section 2 is a reminder 
of what is LG both in physical (Section 2.1) and mathematical
terms (Section 2.2). Section 3 is an elementary reminder of
what is gradient of a function in a finite-dimensional space.
The core of the paper is Section 4.
In Section 4.1, an attempt is made to define the gradient 
in the space of curves, based on analogies 
with the finite-dimensional case. Section 4.2 contains the
main results, Theorems \ref{theorem:main1}, \ref{theorem:main2}
and \ref{theorem:main3}, which directly follow from the definitions
given in Section 4.1
and the Hadamard variational formula. In the last Section 5
possible perspectives are discussed.

\section{Laplacian growth}

\subsection{Physical description}

We begin with a brief description of the LG in physical terms.
For a more detailed exposition see
\cite{RMP,book,MWPT}

In practice, the 2D geometry is
realized in a narrow gap between two parallel
glass plates (the Hele-Shaw cell),
see review \cite{RMP} and references therein.
A compact
plane domain is occupied by a fluid with a 
negligible viscosity (water).
We call it water droplet.
Let its exterior be occupied by a viscous fluid (oil).
The liquids are assumed to be 
incompressible and non-mixing with each other.
The oil/water interface is assumed to be a smooth closed
curve.
Oil is sucked out with a constant rate through a sink
(a pump) placed
at some fixed point (which may be at infinity) while
water is injected
into the water droplet. 

Since 
water has negligible viscosity, pressure
in it is constant. We put it equal to zero without loss of 
generality.
In the oil domain, the local velocity
${\bf v}=(v_x,v_y)$ of the fluid is proportional to
the gradient of pressure $p=p(x,y)$
(Darcy's law):
\beq\label{darcy}
{\bf v}=-\kappa \, \gr p,
\eeq
where $\kappa$ is called
the filtration coefficient. (In general, it may depend 
on coordinates in the plane.)
This means that a viscous fluid in the Hele-Shaw cell flows 
at every point in the direction of the fastest pressure decrease.

In particular, the Darcy law holds on the
outer side of the interface thus
governing its dynamics:
\beq\label{darcy0}
v_n=-\kappa \p_n p.
\eeq
Here $v_n$ is the growth velocity, which is
normal to the interface\footnote{The tangential component 
of the velocity, if any, does not lead to a change in the shape 
of the droplet and therefore has no physical meaning.}, and
$\p_n $ is the normal derivative.
Since the fluids are incompressible ($\gr  {\bf v} =0$),
the Darcy law implies
that pressure $p$
is a harmonic function in the exterior (oil) domain
except at the point where the oil pump is placed.
There is a logarithmic singularity at that point.
A nonzero surface tension means
a pressure jump across the interface
proportional to its local curvature. 
If surface tension is negligibly
small, there is no jump and
it can be considered that the pressure in oil 
is zero along the entire boundary. 

Usually the pump is placed at infinity. 
In this case the idealized problem (with zero surface 
tension) is formulated as follows:
\beq\label{darcy2}
\left \{
\begin{array}{l}
v_n= -\kappa \p_n p,
\\ \\
\Delta p=0 \quad \mbox{in oil},
\\ \\
p=0 \quad \mbox{on the interface},
\\ \\
p(x,y)=-\log \rho +  \ldots \quad \mbox{as $\rho =\sqrt{x^2+y^2}\to
\infty$}.
\end{array}
\right.
\eeq
Here $\Delta =\p_x^2 +\p_y^2$ is the Laplace operator. 
So, from the mathematical perspective, one is dealing with 
a boundary value problem of the Dirichlet type with a 
free moving boundary.

\subsection{Mathematical formulation}

In two dimensions, it is instructive to pass to 
the complex coordinates $z=x+iy$, $\bar z=x-iy$. This allows one 
to apply methods of complex analysis.
To give a mathematical formulation
of LG in these terms, we need some standard notions of 
complex analysis and boundary value problems.

\subsubsection{Conformal map and Green function}

Let us denote the boundary curve (the interface)
by $\gamma$ and let ${\sf D}$ be the compact
(interior) domain bounded by it: $\gamma =\p {\sf D}$ (it is the water
domain).
Its complement (the oil domain) is the 
non-compact domain ${\sf D}'=
\CC \setminus {\sf D}\ni \infty$. Let $w(z)$ be the conformal map
from ${\sf D}'$ onto the exterior of the unit disk such that
$w(\infty )=\infty$ and $w'(\infty )=1/r \in \RR_+$. 
By the Riemann theorem, such a map exists and is unique due to
the normalization condition at infinity.
The Laurent expansion of $w(z)$ around $\infty$ has the form
$$
w(z)=\frac{z}{r} +\sum_{k\geq 0} u_k z^{-k}, \quad z\to \infty .
$$
The real
positive number $r$ is called the (external) conformal radius
of the domain ${\sf D}$.

Next, we need 
the Green function of the Dirichlet boundary value
problem in ${\sf D}'$.
It is a (unique) solution to the equation
$\Delta _z G(z, \zeta )=\Delta_{\zeta}(z,\zeta )=2\pi
\delta^{(2)}(z-\zeta )$ vanishing on $\gamma =\p {\sf D}$,
i.e. $G(z, \xi )=0$ for all $\xi \in \p {\sf D}$. 
Here $\Delta_z =4\p_z \p_{\bar z}$ is the Laplace operator in the
variable $z$.
In other words,
as $\zeta \to z$, the Green function has the logarithmic singularity
$G(z, \zeta )=\log |z-\zeta |+\ldots $, and it is harmonic
elsewhere in ${\sf D}'$. 
The Green function can be expressed through
$w(z)$ by
the formula
\beq\label{G1}
G(z, \zeta ) =\log \left |
\frac{w(z) - w(\zeta )}{1-w(z)\overline{w(\zeta )}}
\right |.
\eeq
In particular,
\beq\label{G2}
G(z, \infty )=-\log |w(z)|.
\eeq

\begin{figure}[t]
\centering{\includegraphics[scale=0.7]{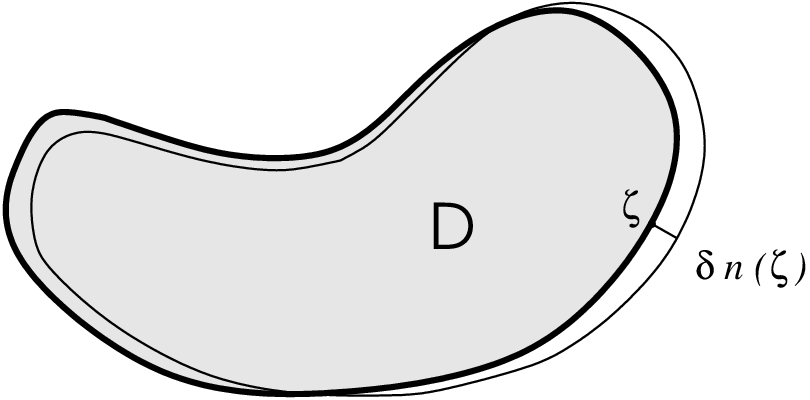}}
\vspace{1cm}
\caption{Infinitesimal variation.}
\label{figure:var}
\end{figure}

We also need the Hadamard variational formula \cite{Schiffer}.
Let $\delta n (\xi )$ be the infinitesimal 
normal displacement of the curve
$\gamma$ at a point $\xi \in \gamma$ 
with the convention that $\delta n>0$ for outward shifts 
and $\delta n<0$ for inward ones (see Fig.~\ref{figure:var}). 
Then the Hadamard formula for variation of the Green function is 
\beq\label{G3}
\delta G(a,b)=\frac{1}{2\pi} \oint_{\gamma}
\p_{n_\xi}G(a, \xi )\p_{n_\xi}G(b, \xi ) \delta n(\xi )
|d\xi |.
\eeq
Here $a,b \in {\sf D}'$ and 
$\p_{n_{\xi}}$ is the normal derivative.
In particular, letting $b\to \infty$, we have:
\beq\label{G4}
\delta \log |w(z)|
=\frac{1}{2\pi} \oint_{\gamma}
\p_{n_\xi}G(z, \xi ) |w'(\xi )|\, \delta n (\xi )
|d\xi |,
\eeq
where the identity $\p_{n_\xi}\log |w(\xi )|=|w'(\xi )|$
for $\xi \in \gamma$ was used.
Further, letting here 
$z\to \infty$, we obtain a formula for variation of the
conformal radius:
\beq\label{G4a}
\delta \log r
=\frac{1}{2\pi} \oint_{\gamma} |w'(\xi )|^2\, \delta n (\xi )
|d\xi |\, .
\eeq
This is a key relation for what follows.

\subsubsection{Laplacian growth in terms of the Green function}

In terms of the conformal map $w(z)$ the solution for pressure
in (\ref{darcy2}) is 
\beq\label{pressure}
p(z)=-\log |w(z)|, \quad z\in {\sf D}'.
\eeq
This function satisfies the last three conditions in (\ref{darcy2}).
Thus the Darcy law reads:
\beq\label{darcy3}
v_n(\xi ) =\p_{n_\xi}\log |w(\xi )|=|w'(\xi )|, \quad \xi \in 
\gamma ,
\eeq
where the filtration coefficient $\kappa$ is set to be 1. 

In a more general case $\kappa$ can be some real-valued 
function on the plane: $\kappa = \kappa (z)$:
\beq\label{darcy4}
v_n(\xi ) =-\kappa (\xi )\, |w'(\xi )|, \quad \xi \in 
\gamma 
\eeq
(For notational
simplicity, we write simply $\kappa (z)$ keeping in mind that
it depends also on $\bar z$.) 
We will refer to this case as an inhomogeneous LG.
We assume that $\kappa (z)>0$ and 
$\displaystyle{\lim_{z\to \infty}
\kappa (z)=1}$.

Another more general case is that of a point-like source
at some finite point $a\in {\sf D}'$ rather than at infinity.
In this case the solution for $p(z)$ is
$p(z)=G(a, z)$, so instead of (\ref{darcy3}) we have:
\beq\label{darcy5}
v_n(\xi ) =-\p_{n_\xi}G( a, \xi ), \quad \xi \in 
\gamma .
\eeq
A ``mixture'' of the last two cases (i.e., inhomogeneous LG with
a source at a finite point) is also possible.

\section{Reminder: what is gradient of a function in 
$\RR^N$}

In order to somehow motivate the definitions 
given in the next section and make them clearer,
we will start with a brief reminder of some elementary
things related to gradient 
in $\RR^N$, presenting in the form convenient 
for the subsequent generalization.

The Euclidean space $\RR^N$ is equipped with
the standard scalar product 
\beq\label{scalar0}
({\bf v}, {\bf g})=\sum_{i=1}^N v_i g_i,
\eeq
where $v_i$ and $g_i$ are components of the vectors 
${\bf v}, {\bf g}\in \RR^N$.
With the scalar product, we can ignore the distinction 
between vectors and covectors.
The squared norm of a vector ${\bf g}$ is
$$
|{\bf g}|^2 =({\bf g}, {\bf g})=\sum_{i=1}^N g_i^2.
$$

A vector field in $\RR^N$ is a function on 
the tangent bundle $T\RR^N \cong \RR^N \oplus \RR^N$, 
or, in other words, an $\RR^N$-valued function
${\bf g}({\bf x})=(g_1({\bf x}), \ldots , g_N({\bf x}))$.
Here ${\bf g}({\bf x})$ should be understood as a tangent
vector at the point ${\bf x}\in \RR^N$ with components
$g_1, \ldots , g_N$.

Let $F({\bf x})=F(x_1, \ldots , x_N)$ be a differentiable 
$\RR$-valued function on $\RR^N$. The variation of the function
along a vector field ${\bf g}$
is defined as
\beq\label{N1}
\delta_{\bf g}F=F ({\bf x} +\epsilon{\bf g}({\bf x}) )-
F ({\bf x} ), \quad \epsilon \to 0_+.
\eeq
By definition, gradient of $F$ is a (co)vector field
$\gr F$ such that
\beq\label{N2}
\delta_{\bf g}F =\epsilon (\gr F , {\bf g}) +O(\epsilon^2)
\eeq
for any vector field ${\bf g}$. 
In what follows, we will not write $O(\epsilon^2)$ in such 
formulas for linear response. 

\begin{remark}\label{remark:velocity}
Applying equation (\ref{N1}) to the vector-valued function ${\bf x}$
and assuming that $\epsilon$ is an infinitesimally 
small time interval, 
$\epsilon =\delta t$, we see that the tangent vector ${\bf g}$ 
makes sense of velocity of a moving particle at the point ${\bf x}$.
\end{remark}

Since 
$$
\delta_{\bf g}F ({\bf x})=\epsilon \sum_{i=1}^N
g_i({\bf x})\frac{\p F ({\bf x})}{\p  x_i},
$$
$\gr F$ in components is written as
$$
\gr F =\Bigl (\frac{\p F}{\p x_1}\, , \ldots , \frac{\p F}{\p x_N}
\Bigr ).
$$

As is well known (and can be easily proved), variation of 
the function $F$ along vector fields ${\bf g}({\bf x})$ with
a fixed norm $|{\bf g}|=R$ is maximum for the vector field
$$
{\bf g}^{(0)}=\frac{R \, \gr F}{|\gr F |}\, .
$$
Informally speaking, when moving along the gradient, 
the increment of the function is maximized.

\section{Vector fields in the space of curves and Laplacian growth}

\subsection{Definitions}

Let ${\cal L}$ be the space of smooth non-self-intersecting
curves (loops) in the plane. For simplicity, they are assumed to be
smooth.
``Points'' of this space (curves)
will be denoted by $\gamma$. 

\begin{definition}
A tangent vector at a point $\gamma \in {\cal L}$ is a (smooth)
real-valued function on the curve $\gamma$.
\end{definition}

\noindent
The informal meaning of this definition is as follows.
Such a function (say, $g(z)$, where $z\in \gamma$) is to be thought
of as representing an infinitesimal deformation of the curve $\gamma$
given by $\delta n(z)=\epsilon g(z)$ with $\epsilon \to 0_+$.
Here $\delta n(z)$ is the normal displacement of the curve at the
point $z$, as before.
We will denote tangent vectors
as ${\bf g}={\bf g}(\gamma )$.

\begin{remark}\label{remark:velocity1}
Similarly to the finite-dimensional case 
(see Remark \ref{remark:velocity}), the tangent vector ${\bf g}(\gamma )$
at the point (a curve) $\gamma$ can be 
thought of as ``velocity of  
deformation of the curve'' meaning that the function $g(\xi )$ is
the normal velocity at the point $\xi \in \gamma$.
\end{remark}

\begin{definition}
A vector field in the space ${\cal L}$ is a family of tangent
vectors such that 
a tangent vector ${\bf g}(\gamma )$ 
is assigned to each point $\gamma \in {\cal L}$ and its dependence
on $\gamma$ is in a certain sense smooth\footnote{Here we will not 
discuss the precise meaning of this smoothness.}.
\end{definition}

\noindent
Let ${\cal F}={\cal F}(\gamma )$ be a functional on the space
${\cal L}$.

\begin{definition}\label{definition:variation}
Variation of a functional ${\cal F}$ along
a vector field ${\bf g}(\gamma )$ in the space ${\cal L}$ 
is
\beq\label{var}
\delta_{\bf g}{\cal F}={\cal F}
(\gamma^{(\epsilon {\bf g})} )-{\cal F}(\gamma ), \quad \epsilon
\to 0_+ \, ,
\eeq
where $\gamma^{(\epsilon {\bf g})}$ is the curve deformed by
the infinitesimal 
shifts $\delta n(\xi )=\epsilon g(\xi )$ at all
points $\xi \in \gamma$.
\end{definition}

\noindent
This is an analogue of equation (\ref{N1}).
In the notation of this definition,
the Hadamard formula (\ref{G3}) and its corollaries 
(\ref{G4}), (\ref{G4a}) read:
\beq\label{G3a}
\delta_{\bf g} G(a,b)=\frac{\epsilon}{2\pi} \oint_{\gamma}
\p_{n_\xi}G(a, \xi )\p_{n_\xi}G(b, \xi ) g(\xi )
|d\xi |,
\eeq
\beq\label{G4b}
\delta_{\bf g} \log |w(z)|
=\frac{\epsilon}{2\pi} \oint_{\gamma}
\p_{n_\xi}G(z, \xi ) |w'(\xi )|\, g(\xi )
|d\xi |,
\eeq
\beq\label{G4c}
\delta_{\bf g} \log r
=\frac{\epsilon}{2\pi} \oint_{\gamma} |w'(\xi )|^2\, g(\xi )
|d\xi |\, .
\eeq
These are our key relations.

Now we are ready to define a scalar product in the tangent space
at a ``point'' $\gamma \in {\cal L}$.

\begin{definition}\label{definition:scalar1}
Let ${\bf g}_1, \, {\bf g}_2$ be two tangent vectors at
$\gamma \in {\cal L}$. The scalar product is defined as
follows:
\beq\label{scalar}
({\bf g}_1, \, {\bf g}_2)_{\gamma} =\frac{1}{2\pi}
\oint_{\gamma} g_1(\xi )g_2(\xi ) \, |w'(\xi )| \, |d\xi |.
\eeq
\end{definition}

\begin{remark}
The factor $|w'(\xi )|$ under the integral has 
appeared due to the following considerations. The ``components''
of a continually infinite-dimensional tangent 
vector ${\bf g}$ should be 
labeled by some ``index'', like $i$ in (\ref{scalar0}). However,
the variable $\xi$ itself is not a good candidate for this because it
depends on the very ``point'' $\gamma$. If we set $w(\xi )=
e^{i\varphi}$, then $ |w'(\xi )| \, |d\xi |=d\varphi$ and
the variable $\varphi$ serves then as a universal 
$\gamma$-independent ``index''. A related reason is that
$|w'(\xi )| \, |d\xi |$ is dimensionless. So, if the
dimension of $z$ is $[{\rm length}]$, then the dimension
of the bilinear scalar product is $[{\rm length}]^2$, which is
quite desirable and natural.
\end{remark}

By analogy with the finite-dimensional case, we give the
following definition of the gradient\footnote{We denote it as 
$\Gr$ to avoid confusion with its finite-dimensional counterpart.}.

\begin{definition}\label{definition:Grad}
Let ${\cal F}$ be a functional on ${\cal L}$. The gradient
$\Gr {\cal F}$ is a vector field in ${\cal L}$ such that
\beq\label{Grad1}
\delta_{\bf g}{\cal F}=\epsilon (\Gr {\cal F}, {\bf g})_{\gamma}
\eeq
for any vector field ${\bf g}={\bf g}(\gamma )$.
\end{definition}
This is the analogue of (\ref{N2}).

\subsection{Laplacian growth as a gradient descent}

\begin{theorem}\label{theorem:main1}
The vector field in ${\cal L}$ 
corresponding to the deformation of curves according 
to Laplacian growth with a source at infinity (according to Darcy’s law)
is $-\Gr {\cal F}$, where ${\cal F}=\log (1/r)$ ($r$ is the 
external conformal radius of the compact domain bounded by the curve).
\end{theorem}

\noindent
{\it Proof.} This statement is an almost trivial consequence
of the definitions and the Hadamard formula. 
Namely, from (\ref{G4a}), 
(\ref{scalar}) and (\ref{Grad1}) we see that the vector field
$-\Gr (\log (1/r))$ is just $|w'(\xi )|$, i.e., the 
corresponding deformation
of the curve is $\delta n(\xi )=\epsilon |w'(\xi )|$, which is
Darcy's law of Laplacian growth with a source at infinity
(see (\ref{darcy3})).
\square

\noindent
Informally speaking,  
Laplacian growth means 
the fastest possible increase in the conformal radius of a growing 
domain.

\begin{remark}\label{remark:electrostatic}
The functional ${\cal F}=\log (1/r)$ has the following electrostatic
interpretation. It is the (regularized) energy of the electric
field created by the conductor ${\sf D}$ carrying the unit charge:
$$
{\cal E}({\sf D})
=\frac{1}{2\pi}\int_{\CCC \setminus {\sf D}}
\Bigl |\p_z G(z, \infty )\Bigr |^2 \, d^2 z.
$$
This integral logarithmically diverges at infinity. The regularization
consists in subtracting the similar integral for the unit disk ${\sf U}$,
then
$$
{\cal E}({\sf D})-{\cal E}({\sf U})=\log (1/r).
$$
\end{remark}

The two generalizations of LG (\ref{darcy4}) and (\ref{darcy5})
have a gradient nature, too. For the inhomogeneous LG we have
the following analogue of Theorem \ref{theorem:main1}:

\begin{theorem}\label{theorem:main2}
The vector field in ${\cal L}$ 
corresponding to the deformation of curves according 
to the inhomogeneous LG with the space-dependent filtration
coefficient $\kappa (z)$ and source at infinity
is $-{\bf Grad}^{(\kappa )}{\cal F}$, where ${\cal F}=\log (1/r)$
and ${\bf Grad}^{(\kappa )}$ is the gradient defined as in Definition
\ref{definition:Grad} with respect to the scalar product
\beq\label{scalar1}
({\bf g}_1, \, {\bf g}_2)_{\gamma}^{(\kappa )} =\frac{1}{2\pi}
\oint_{\gamma} g_1(\xi )g_2(\xi ) \, 
\frac{|w'(\xi )| \, |d\xi |}{\kappa (\xi )}.
\eeq
\end{theorem}

\noindent
The proof almost literally repeats that of 
Theorem \ref{theorem:main1}. So, in this case 
the functional is the same but
the gradient is defined with respect to another scalar product.

For the LG with a source at a finite point $a$ the scalar product
is the same as in Theorem \ref{theorem:main1} but the functional is
different.

\begin{theorem}\label{theorem:main3}
The vector field in ${\cal L}$ 
corresponding to the deformation of curves according 
to Laplacian growth with a source at a finite
point $a$ 
is $-\Gr {\cal F}_a$, where ${\cal F}_a=\log |w (a)|$.
\end{theorem}

\noindent
The proof is similar to that of Theorem \ref{theorem:main1}.
The only difference is that instead of corollary (\ref{G4a}) of the
Hadamard formula one should use more general equation (\ref{G4}).

\section{Discussion}

Let us discuss possible generalizations.

First, we note that ``gradient decent'' interpretation of LG
should be easily generalizable to the 3D case 
(and even to any dimensions). Although a proper definition of
the conformal radius in the 3D case is problematic, an analog
of Theorem \ref{theorem:main1} should be valid with the functional
defined as in Remark
\ref{remark:electrostatic}, i.e., as the electrostatic energy
of the conducting compact domain in 3D carrying unit charge.

Second, 
it would be interesting to find out whether some other 
boundary value problems with a free moving boundary 
allow for a similar gradient interpretation, such as the LG
with non-zero surface tension, as well as the problem of 
Stokes flows. Note that in the latter case, an infinite set of 
conserved quantities is
known to exist \cite{CHK97}.

The most interesting (albeit still hypothetical) prospect
is related to a possible introduction of a random noise,
similar to how this is usually done in finite-dimensional 
dynamical systems realized as gradient descent.
Namely, one can try to formulate a stochastic version of LG
as an infinite-dimensional generalization of the standard problem of a
dissipative system in a potential field 
with a random force (Gaussian noise), which is described by the 
Langevin equation. As is known, the Langevin equation leads,
with the help of It\^o calculus  (see, e.g. \cite{Klebaner}), 
to the Fokker-Planck equation
for the probability density. 

Coming back to the finite-dimensional case as a source of motivation,
the stochastic version of the deterministic 
``gradient descent'' system in a potential $U$,
$
\p_t{\bf x}=-\gr U({\bf x})
$,
is
\beq\label{d1}
\p_t{\bf x}=-\gr U({\bf x}) + {\bf f} (t),
\eeq
where ${\bf f}(t)=\{f_1(t), \ldots , f_N(t)\}$ is a 
Gaussian random force. Because of the random force, the trajectories
have some probabilistic distribution. It can be described by
a function $P({\bf x}, t)$, which is the probability density
for the system to have coordinates ${\bf x}$ at the time $t$. 
The Fokker-Planck equation for $P({\bf x}, t)$ reads
\beq\label{d2}
\p_t P=\eta \Delta P + (\gr U , \gr P)+\Delta U \cdot P ,
\eeq
where $\eta$ is a parameter characterizing strength of the random
force and $\Delta$ is the Laplace operator in $\RR^N$.

The idea is to find a proper analogue of the Fokker-Planck
equation in the space ${\cal L}$. Solving this equation 
would make it possible to find the most likely 
shape of the interface over long time periods.
However, there is an obvious difference from the standard 
formulation of problems involving gradient descent with noise.
In the latter case, the potential usually has a lot of local minima, 
and the task is to reach the lowest one. On the contrary, in the case of 
LG the functional $\log (1/r)$ decreases indefinitely to 
minus infinity as the area of the growing domain increases.
Therefore, the problem needs to be posed in a different way.
Indeed, of greatest interest is formation of a fractal structure of the interface boundary over long time scales, and for it to become clearly visible, the process needs to be normalized to a constant area. 
Then, the numerous branching fingers will become smaller 
and smaller in size, while their number increases, 
which ultimately may lead to a fractal structure.
To achieve this, simultaneously with the growth of the 
domain, according to Darcy’s law, the domain must be 
scaled (compressed) so that the area remains constant.
This will certainly modify Darcy's law.

\section*{Acknowledgments}
\addcontentsline{toc}{section}{Acknowledgments}

I am grateful to O. Alekseev and M. Mineev-Weinstein for numerous
discussions of problems related to Laplacian growth.

This article is an output of the research project 
(HSE-BR-2025-84) implemented as a part of the 
Basic Research Program at the National Research University 
Higher School of Economics (HSE University).

\end{document}